\documentclass[11pt]{article}
\usepackage{amsmath,amsfonts,amssymb,amsthm}
\usepackage{fullpage}
\usepackage{subfigure} \usepackage{epstopdf}
\usepackage{graphicx}
\usepackage{color}
\usepackage[sans]{dsfont}
\usepackage{mathtools}
\usepackage{hyperref}
\usepackage{cleveref}
\usepackage{tcolorbox}
\usepackage{multirow}
\usepackage[shortlabels]{enumitem}
\usepackage{todonotes}
\usepackage{xspace}
\tcbuselibrary{skins,breakable}
\tcbset{enhanced jigsaw}
\usepackage[font=small,labelfont=bf]{caption}

\usepackage{fullpage}
\usepackage[utf8]{inputenc}
\usepackage[noadjust]{cite}

\usepackage{graphicx}
\usepackage{amsmath,amsthm,amssymb}
\usepackage{algorithm}
\usepackage{algpseudocode}
\usepackage{multirow}
\usepackage{makecell}
\usepackage{algpseudocode}

\usepackage{thm-restate,color,xspace}
\usepackage{comment}
\usepackage{thmtools}
\usepackage{xcolor}
\usepackage{nameref}
\usepackage{array}

\usepackage{todonotes}

\makeatletter
\newcommand{\pushright}[1]{\ifmeasuring@#1\else\omit\hfill$\displaystyle#1$\fi\ignorespaces}
\newcommand{\pushleft}[1]{\ifmeasuring@#1\else\omit$\displaystyle#1$\hfill\fi\ignorespaces}
\makeatother

\newtheorem{theorem}{Theorem}
\newtheorem{lemm}{Theorem}[section]

\newtheorem{lemma}[theorem]{Lemma}

\newtheorem{conjecture}[theorem]{Conjecture}
\theoremstyle{definition}
\newtheorem{claim}[theorem]{Claim}

\newtheorem{definition}[lemm]{Definition}

\newcommand{\polylog}{\text{\normalfont polylog }}

\newcommand{\eps}{{\varepsilon}}

\newcommand{\dist}{\textnormal{\textsf{dist}}}

\newcommand{\mst}{\texttt{MST}}

\newcounter{note}

\begin{document}

\title{Multiplicative Spanners in Minor-Free Graphs}

\begin{titlepage}
\date{}

%ESA 25 uses double blind reviewing
\author{Greg Bodwin\thanks{University of Michigan, MI, USA. Email: {\tt bodwin@umich.edu}.} \and Gary Hoppenworth\thanks{University of Michigan, MI, USA. Email: {\tt garytho@umich.edu}.} \and Zihan Tan\thanks{Rutgers University, NJ, USA. Email: {\tt zihantan1993@gmail.com}.} }

\maketitle

\thispagestyle{empty}

\begin{abstract}
In FOCS 2017, Borradaille, Le, and Wulff-Nilsen addressed a long-standing open problem by proving that minor-free graphs have light spanners.
Specifically, they proved that every $K_h$-minor-free graph has a $(1+\eps)$-spanner of lightness $O_{\eps}(h \sqrt{\log h})$, hence constant when $h$ and $\eps$ are regarded as constants.

We extend this result by showing that a more expressive size/stretch tradeoff is available.
Specifically: for any positive integer $k$, every $n$-node, $K_h$-minor-free graph has a $(2k-1)$-spanner with sparsity
$$O\left(h^{\frac{2}{k+1}} \cdot \polylog h\right),$$
and a $(1+\eps)(2k-1)$-spanner with lightness
$$O_{\eps}\left(h^{\frac{2}{k+1}} \cdot \polylog h \right).$$
We further prove that this exponent $\frac{2}{k+1}$ is best possible, assuming the girth conjecture.
At a technical level, our proofs leverage the recent improvements by Postle (2020) to the remarkable \emph{density increment theorem} for minor-free graphs.
\end{abstract}

\end{titlepage}

%\clearpage

\setcounter{page}{1}
\section{Introduction}

Given an input graph $G$, a \emph{$k$-spanner} is a sparse subgraph $H$ whose shortest path distances agree with those of $G$ up to a small multiplicative factor $k$, called the \emph{stretch}.
Spanners were introduced by Peleg and Schaffer \cite{PS89} after being used implicitly in prior work in distributed computing \cite{PU87, PU89}, and have since found widespread application in network design, graph algorithms, etc (see survey \cite{survey}).

A primary goal of research on spanners is to optimize the worst-case tradeoff between stretch and \emph{size}.
There are two popular ways to measure the size of a spanner:
\begin{itemize}
\item The first is by its \emph{sparsity}, $|E(H)|/n$ (where $n$ is the number of vertices in the input graph).
This measure arises e.g.\ when the goal is to store the spanner in few bits of memory, or when running algorithms on the spanner whose runtime is sensitive to the spanner sparsity.

\item The second is by its total weight $w(H)$.
Since $w(H)$ can generally be unbounded, it is popular to normalize the spanner weight by the weight of a minimum spanning tree (MST) of the input graph: the lightness of a spanner $H$ of a graph $G$ is the quantity
$$\ell(H \mid G) := \frac{w(H)}{w(\mst(G))}.$$
\end{itemize}

\subsection{Sparsity and Lightness in General Graphs}

The lightness of a graph is always at least as large as its sparsity, and in general it could be much higher.
However, in some sense this difference can be bounded.
Let us restrict attention to $n$-node \emph{metric graphs}, in which every edge is the unique shortest path between its endpoints, which is the natural setting for spanner problems (i.e., non-metric edges can be removed without affecting the distances in $G$).
Here, the bounds are:
\begin{theorem} [Folklore]
Every $n$-node metric graph $G$ has sparsity $O(n)$ and lightness $O(n^2)$.
\end{theorem}

For the upper bound on lightness, note that the graph $G$ has at most $n^2$ edges, each of weight at most $w(\mst(G))$.  For the lower bound, consider an $n$-node path with edges of weight $1$, augmented with all additional edges where the weight of $(u, v)$ is set to their distance along the path minus $\eps$.
Although this bound is tight, it has been known since the seminal work of Khuller, Raghavachari, and Young \cite{KRY95} that one can improve lightness to $O(n)$, matching sparsity, by paying a $(1+\eps)$ factor in stretch:

\begin{theorem} [\cite{KRY95}]
For any $\eps > 0$, every $n$-node graph $G$ has a $(1+\eps)$-spanner with lightness $O_{\eps}(n)$.\footnote{We use $O_{\eps}$ notation to hide polynomial dependencies on $\eps^{-1}$.}
\end{theorem}

Once this initial upper bound is established, it is then natural to ask whether one can pay an even higher stretch to improve the $O(n)$ sparsity and lightness dependence.
After a long line of work \cite{ADDJS93, CDNS92, ENS15, CW18, FS20, LS23, Bodwin25, BF25}, especially on the side of lightness, the following tradeoffs is now known:

\begin{theorem} [\cite{ADDJS93, CW18, LS23, Bodwin25}] \label{thm:genspan}
For any $\eps > 0$ and positive integer $k$, every $n$-node graph has a $(2k-1)$-spanner with sparsity $O(n^{1/k})$ and a $(1+\eps)(2k-1)$-spanner with lightness $O_{\eps}(n^{1/k})$.
\end{theorem}

This sparsity bound is known to be best possible, assuming the girth conjecture \cite{girth}.
For lightness, the $n^{1/k}$ dependence is tight assuming the girth conjecture as well.
The exact (hidden) dependence on $\eps$ remains open, but it was recently proved that the $\eps$-dependence is not entirely removable \cite{BF25}.

\subsection{Sparsity and Lightness in Minor-Free Graphs}

An important special case for spanners is that where the input is $K_h$-minor-free; this arises as a technical tool in approximation algorithms for the Traveling Salesman Problem \cite{RS98, AGKKW98, BDT14, GS02, DHK11}, and in some sense it generalizes the well-studied setting of Euclidean input graphs (see book \cite{NS07}).
A $K_h$-minor-free graph is already guaranteed to be quite sparse, and therefore comparably light:

\begin{theorem} [\cite{Kostochka82, AKS23}]
Every $n$-node, $K_h$-minor-free metric graph has sparsity $O(h \sqrt{\log h})$ and lightness $O(n h \sqrt{\log h})$.\footnote{For a gradual exposition of the proof and development of ideas in this theorem, we recommend the excellent recent blog post by Hung Le \cite{hungblog}.}
\end{theorem}

Is is natural to ask, like for general graphs, whether one can remove the $n$-dependence from the lightness bound by paying an extra $(1+\eps)$ factor in the stretch.
This was a long-standing open question: it was first asked in 2002 by Grigni and Sissokho \cite{GS02}, and then conjectured to be possible in 2012 by Grigni and Hung \cite{GH12} (who showed that $(1+\eps)$ stretch can be used to improve the dependence to $\log n$).
The result was finally achieved in a remarkable FOCS 2017 paper by Borradaille, Le, and Wulff-Nilsen:

\begin{theorem} [\cite{BLW17}] \label{thm:lightmfspan}
For any $\eps > 0$, every $n$-node, $K_h$-minor-free graph has a $(1+\eps)$-spanner with lightness $O_{\eps}(h \sqrt{\log h})$.
\end{theorem}

%\ztnote{the reference should be \cite{BLW17}? Also, their result seems to be about $(1+\eps)$-spanners (without dependence on $k$), and also we may want to say $O_{\eps}(h \sqrt{\log h})$ to include dependence on $\eps$?}

Analogous to the work on general graphs, a natural next step after this theorem is to investigate the \emph{full} tradeoff between stretch and sparsity/lightness for minor-free graphs.
That is the subject of the current paper.
We prove the following new bounds:

\begin{theorem} [Main Result, Upper Bounds] \label{thm:mainupper}
For any $\eps > 0$ and positive integer $k$, every $n$-node, $K_h$-minor-free graph has a $(2k-1)$-spanner with sparsity
$$O\left(h^{\frac{2}{k+1}} \cdot \polylog h\right),$$
and a $(1+\eps)(2k-1)$-spanner with lightness
$$O_{\eps}\left(h^{\frac{2}{k+1}} \cdot \polylog h \right).$$
\end{theorem}

The exponent $\frac{2}{k+1}$ in Theorem \ref{thm:mainupper} might look surprising; at first, one might guess that the exponent would be $\frac{1}{k}$ to match Theorem \ref{thm:genspan}.
However, we further prove that this is indeed the correct exponent for minor-free graphs:

\begin{theorem} [Main Result, Lower Bounds]
Assuming the girth conjecture \cite{girth}, there are $n$-node $K_h$-minor-free graphs for which any $(2k-1)$-spanner has sparsity and lightness
$$\Omega\left( h^{\frac{2}{k+1}} \right).$$
\end{theorem}

Our upper and lower bounds differ by $\polylog h$ factors.
It seems unlikely that these factors can be removed entirely from the upper bound, since we know that a $\sqrt{\log h}$ factor is truly necessary in the sparsity and lightness bounds for minor-free graphs.
However, it would be interesting to have this $\polylog h$ dependence degrade with the stretch.
We conjecture that the following improvement to our bounds will be possible:

\begin{conjecture} \label{cjt:betterbound}
For any $\eps > 0$ and positive integer $k$, every $n$-node, $K_h$-minor-free graph has a $(2k-1)$-spanner with sparsity
$$O\left(\left(h \cdot \polylog h\right)^{\frac{2}{k+1}}\right),$$
and a $(1+\eps)(2k-1)$-spanner with lightness
$$O_{\eps}\left(\left(h \cdot \polylog h\right)^{\frac{2}{k+1}} \right).$$
\end{conjecture}

The advantage of this conjecture over our main result would be that it would allow spanners of truly constant sparsity/lightness (even as a function of $h$), in exchange for stretch $O(\log h)$.
This would be a satisfying analog to the important special case of spanners with constant sparsity/lightness and stretch $O(\log n)$ for general graphs.

\subsection{Techniques: A Density Increment Theorem for Minor-Free Graphs}
A \emph{density increment theorem} for minor-free graphs is a result that guarantees that every minor-free graph has a relatively dense subgraph.
Density increment theorems for minor-free graphs have emerged as a powerful tool in recent progress on Hadwiger's Conjecture \cite{postle2020even, DP25, NPS23}.  
This paper makes the point that they may be far more broadly applicable to algorithmic and combinatorial analysis of minor-free graphs as well.

For intuition, let us revisit the bound of $O(h \sqrt{\log h})$ on the sparsity of minor-free graphs.
One can get a nearly-matching lower bound by considering the disjoint union of $\frac{n}{h-1}$ cliques of size $h-1$ each: this has sparsity $\Omega(h)$, and is clearly $K_h$-minor-free.
The actual tight lower bound is only slightly more involved than that; it is a disjoint union of many carefully-designed small graphs that are each nearly a clique.

These constructions based on disjoint unions will contain ``clustered'' parts of the graph, i.e., very small subgraphs whose average degree is roughly the same as the original graph.
The Density Increment Theorem says, in some sense, that this is not a coincidence: \emph{all} nearly-tight lower bound graphs must have this structure.

\begin{theorem}[Density Increment Theorem  \cite{postle2020even}, reformulated]\footnote{The cited work has been withdrawn from arXiv, since the goal of this theorem in the paper was to prove an $O(t\log\log^6 t)$ bound on Hadwiger's conjecture and an improved bound of $O(t\log\log t)$ was proved by the author and Delcourt in a subsequent work \cite{delcourt2025reducing} via a similar approach. However, the stated Density Increment Theorem still holds. See also the remark in \cite{delcourt2025reducing} (the paragraph after Theorem 2.2).} \label{thm: density increment} 
Let $G$ be a $K_h$-minor-free graph with average degree $d$. Then $G$ contains a subgraph on at most $(h^2/d) \cdot \text{\normalfont polylog } h$ vertices, with average degree at least $\frac{d}{\text{\normalfont polylog } h}$.
\end{theorem}

For analyzing spanner sparsity, the density increment theorem yields a remarkably simple proof.
Our strategy is to use the standard greedy spanner algorithm \cite{ADDJS93} to construct a $K_h$-minor-free spanner $H$ with high girth.
We can then use the density increment theorem to pass to a very small subgraph $H' \subseteq H$ of comparable sparsity, which still has high girth.
The standard Moore bounds, which limit the maximum possible density of a high-girth graph, then yield our desired theorem when applied to $H'$.
The gap of $\polylog h$ between Theorem \ref{thm:mainupper} and Conjecture \ref{cjt:betterbound} comes entirely from the $\polylog h$ loss in the density increment theorem itself.

For analyzing spanner lightness, the density increment theorem is still useful, but the proof is necessarily more involved.
We again start by using the greedy algorithm to construct a $K_h$-minor-free spanner $H$ with high \emph{weighted} girth (see Definition \ref{def:wtdgirth}).
However, we cannot simply use the density increment theorem to pass to a small subgraph $H'$: although we could bound the lightness of $H'$, we do not have a good way to compare $\mst(H')$ and $\mst(H)$, so this would not reveal much about the lightness of $H$.
Instead, we need to dip into the techniques of Borradaille, Le, and Wulff-Nilsen \cite{BLW17} for proving the initial bound on minor-free spanner lightness in Theorem \ref{thm:lightmfspan}.

\subsection{Other Related Work}

Spanners have been extensively studied in Euclidean graphs, where the nodes are embedded in Euclidean space, there are all possible edges among nodes, and the weight of each edge $(u, v)$ set to the Euclidean distance $w(u, v) := \|u - v\|_2$.
We refer to the book by Narasimhan and Smid for a thorough account \cite{NS07}.
The following near-optimal tradeoffs were recently determined by Le and Solomon \cite{LS22}:

\begin{theorem} [Constant Sparsity/Lightness Spanners for Euclidean Graphs~\cite{LS22}] \label{thm:euclidean}
For any constant $\eps > 0$ and constant integer $d \ge 2$, every $d$-dimensional Euclidean graph has a spanner with $\cdot (1+\eps)$ stretch, sparsity $O\left( \eps^{-d+1} \right)$, and lightness $O\left( \eps^{-d} \log \eps^{-1}\right)$.
These tradeoffs are best possible, up to the factor of $\log \eps^{-1}$ in lightness.
\end{theorem}

One can again ask whether it is possible to pay more stretch in exchange for a better dependence on $\eps$ or $d$.
This question has been studied, but so far remains open.
Aronov et al.~\cite{ADCGHSV08} showed that one can pay a \emph{much} larger stretch of $O(n/k)$ in exchange for an \emph{extremely} sparse spanner with number of edges $|E(H)| \le n-1+k$.
Work of Filtser and Neiman~\cite{FN22} (c.f.\ Remark 1) implied that one can pay $O(t)$ stretch for sparsity and lightness $\exp(d/t^2) \cdot \polylog n$.
However, a general tradeoff remains elusive.

Euclidean graphs are naturally generalized by \emph{doubling metrics}, where analogous sparsity/lightness bounds for $(1+\eps)$-spanners, e.g.~\cite{Gottlieb15, BLW19b, FS20}.
Extending these results to higher stretch would be interesting as well.

\section{Preliminaries}

\paragraph{Graph Notation.}  
Given a graph $G$ and an edge $e \in E(G)$, we let $G-e$ denote the graph obtained by removing $e$ from $G$. We denote the weight of an edge $e$ in $G$ by $w(e)$. Given a pair of vertices $s, t \in V(G)$, we let $\dist_G(s, t)$ denote the weighted distance between $s$ and $t$ in graph $G$.

Our proofs will require a standard theorem in the graph spanner literature that bounds the number of edges in large girth graphs. 

\begin{theorem}[Moore bounds]
\label{thm:moore_bounds}
    For all positive integers $k, n$, every $n$-vertex graph with girth $>2k$ has at most $O(n^{1+1/k})$ edges.
\end{theorem}

Erd\H{o}s conjectured that this upper bound is tight; this is frequently referred to as the girth conjecture. 

\begin{conjecture}[Erd\H{o}s girth conjecture] \label{conj:girth}
    The Moore bounds are tight. That is, for every positive integer $k$ and sufficiently large positive integer $n$,  there exists an $n$-vertex graph with girth $>2k$ and $\Omega(n^{1+1/k})$ edges.
\end{conjecture}

% \begin{theorem}[Density Increment Theorem  \cite{postle2020even}, original]
% \label{thm:density_increment_old}
% Let $G$ be a graph with average degree $d$. Then for any $D>d$, either 
% \begin{itemize}
% \item $G$ contains a $K_{\Omega(D/\sqrt{\log D})}$-minor, or
% \item $G$ contains a subgraph on at most $(D^2/d)\cdot \text{\normalfont polylog}(D/d)$ vertices, with average degree at least $\frac{d}{\text{\normalfont polylog} (D/d)}$.
% \end{itemize} 
% \end{theorem}

% % \begin{proof}[Proof of Theorem \ref{thm: density increment}]  
% \end{proof} 

\section{Sparse Multiplicative Spanners in Minor-Free Graphs}
\label{sec: size}

In this section we will prove that $K_h$-minor-free graphs have sparse multiplicative spanners. 

\begin{restatable}[Minor-Free Spanner Sparsity]{theorem}{sparsespanner} 
For all positive integers $k, h, n$, every $K_h$-minor-free weighted graph on $n$ vertices admits a $(2k-1)$-multiplicative spanner with number of edges
$$O\left(n \cdot h^{\frac{2}{k+1}} \cdot \text{\normalfont polylog }h \right) $$
Moreover, this bound is tight up to $\text{\normalfont polylog }h$ factors under the   girth conjecture. 
\label{thm:sparse_spanner}
\end{restatable}

Let $G$ be an $n$-vertex $K_h$-minor-free weighted graph, and let $k$ be a positive integer. We will construct a $(2k-1)$-spanner $H \subseteq G$ using a greedy strategy that is standard in the area. We first present the construction of $H$, and then analyze its correctness and size. 

\paragraph{Construction of $H$.}  We may assume without loss of generality that for each edge $(u, v) \in E(G)$, $w((u, v)) = \dist_G(u, v)$ by removing edges from $G$ that do not satisfy this condition.
Initially, let $H \gets (V, \emptyset)$.  For each edge $(u, v) \in E$ in order of non-decreasing weight, if $\dist_H(u, v) > (2k-1)  \cdot \dist_G(u, v)$, then add edge $(u, v)$ to $H$. After iterating through all edges in $E$, return $H$. 

We  quickly verify that $H$ is indeed a multiplicative spanner of $G$. 

\begin{claim}[Correctness]
    $H$ is a $(2k-1)$-spanner of $G$. 
    \label{clm:correct}
\end{claim}
\begin{proof}
    This proof follows from a standard argument that can be attributed to~\cite{ADDJS93}. For each edge $(u, v) \in E(G)$, we know that either $\dist_H(u, v) \leq (2k-1) \cdot \dist_G(u, v)$ or $(u, v) \in E(H)$. Since $w((u, v)) = \dist_G(u, v)$, we conclude that for each edge $(u, v) \in E(G)$, $\dist_H(u, v) \le (2k-1) \cdot \dist_G(u, v)$. Combining this fact with the triangle inequality, it follows that for all $s, t \in V(G)$, $\dist_H(s, t) \leq (2k-1) \cdot \dist_G(s, t)$. 
\end{proof}

\paragraph{Size Analysis of $H$.}

Next we prove our desired upper bound on the number of edges in $H$. We begin with the following observation about $H$.

\begin{claim}
    $H$ is a $K_h$-minor-free graph with girth $>2k$. 
    \label{clm:size}
\end{claim}
\begin{proof}
Since $G$ is $K_h$-minor-free, $H$ is as well. We can prove $H$ has girth $>2k$ using a standard argument that we can attribute to~\cite{ADDJS93}. Let $C$ be a simple cycle in $H$. We will show $|C| > 2k$. Consider the edge $e \in C$ with highest weight $w(e)$. The path between the endpoints of edge $e$ in $C-e$ has length at most $(|C|-1) \cdot w(e)$. By the construction of $H$, we would not add edge $e$ to $H$ if  $(|C|-1) \cdot w(e) \le (2k-1) \cdot w(e)$. We conclude that $|C| > 2k$. 
\end{proof} 

Now applying the density increment theorem, we obtain the following general theorem about the number of edges in minor-free graphs with large girth.

\begin{lemma}
    For all positive integers $k, h, n$, every $n$-vertex $K_h$-minor-free graph with girth $>2k$ has at most $O(n \cdot h^{\frac{2}{k+1}} \cdot \text{\normalfont polylog } h)$ edges.
    \label{lem:size}
\end{lemma}
\begin{proof}
    Let $H$ be an  $n$-vertex $K_h$-minor-free graph with girth $>2k$ and average degree $d$. Then by the Density Increment Theorem, there exists a subgraph $H' \subseteq H$ with at most $|V(H')| \leq (h^2/d) \cdot \text{\normalfont polylog } h$ vertices and  at least $|E(H')| \geq |V(H')| \cdot \frac{d}{\text{\normalfont polylog }h}$. Subgraph $H'$ has girth $>2k$, so by Theorem \ref{thm:moore_bounds},
    $$
    \frac{d}{\text{\normalfont polylog }h} \cdot |V(H')| \leq |E(H')| \leq O(|V(H')|^{1+1/k}).
    $$
    Plugging in $|V(H')| \leq (h^2/d) \cdot \text{\normalfont polylog } h$ yields the inequality
    \begin{align*}
        \frac{d}{\text{\normalfont polylog }h}  & \leq (h^2/d)^{1/k} \cdot \text{\normalfont polylog }h,          
    \end{align*}
    and solving for $d$ gives us
    $$
    d \leq h^{2/(k+1)}  \cdot \text{\normalfont polylog }h. 
    $$
\end{proof}

We can now complete the proof of Theorem \ref{thm:sparse_spanner}.

\sparsespanner*
\begin{proof}
    By Claim \ref{clm:correct}, Claim \ref{clm:size}, and Lemma \ref{lem:size},  subgraph $H$  is a $(2k-1)$-spanner of $G$ of size $|E(H)| = O(n \cdot h^{2/(k+1)} \cdot \text{\normalfont polylog }h)$. What remains is to prove that this upper bound is tight up to $\text{\normalfont polylog } h$ factors under the girth conjecture (Conjecture \ref{conj:girth}).

    Formally, we will show that for all positive integers $k, h,$ and for sufficiently large  $n$, there exists an $n$-vertex $K_h$-minor-free graph $G$ that does not admit a $(2k-1)$-spanner $H$ of $G$ of size $|E(H)| = \Omega(n \cdot h^{2/(k+1)})$. 
    If $h = O(1)$, then we can let $G$ be a tree. Every spanner $H$ of $G$ must include every edge of the tree, so $|E(H)| = \Omega(n)$, as desired. Otherwise, if $h$ is sufficiently large, then by the girth conjecture there exists a graph $G^*$ with  $|V(G^*)| = \Theta(h^{2k/(k+1)})$ vertices,  $|E(G^*)| \geq \Omega(h^2)$ edges, and girth $>2k$. Moreover, we can guarantee that $|E(G^*)| < {h \choose 2}$ by removing edges from $G^*$.

    We compose our final lower bound graph $G$ by unioning $ n / |V(G^*)|  = \Theta\left( \frac{n}{h^{2k/(k+1)}}\right)$ disjoint copies of graph $G^*$. Graph $G$ will have $\Theta(n)$ vertices (we can ensure $|V(G)| = n$ by deleting or adding vertices as needed). Graph $G$ has girth $>2k$ and is $K_h$-minor-free, since every connected component of $G$ has fewer than ${h \choose 2}$ edges. Since graph $G$ has girth $>2k$, every $(2k-1)$-spanner $H$ of $G$ must contain every edge in $G$, i.e., $H = G$. We conclude that
    $$
    |E(H)| = |E(G)| = \Theta\left( \frac{n}{h^{2k/(k+1)}}\right) \cdot h^2  = \Theta(n\cdot h^{2/(k+1)}),
    $$
    as claimed.
\end{proof}

\section{Light Multiplicative Spanners in Minor-Free Graphs}

In this section we will prove that $K_h$-minor-free graphs have multiplicative spanners with small lightness.

\begin{restatable}[Minor-Free Spanner Lightness]{theorem}{lightspanner} 
For all positive integers $k, h, n$, and sufficiently small $\varepsilon > 0$, every $K_h$-minor-free weighted graph on $n$ vertices admits a $(1+\varepsilon)(2k-1)$-multiplicative spanner with lightness $$O\left(h^{\frac{2}{k+1}} \cdot \text{\normalfont polylog}(h) \cdot \frac{1}{\eps^3} \log \frac{1}{\eps}\right).$$ 
\label{thm:light_spanner}
\end{restatable}

Our proof will directly follow the proof of Borradaile, Le, and Wulff-Nilsen~\cite{BLW17} that $K_h$-minor-free graphs have $(1+\eps)$-spanners with lightness $O\left(h \sqrt{\log h} \cdot \frac{1}{\eps^3} \log \frac{1}{\eps} \right)$. In fact, our proof will only have two differences from their proof.
\begin{itemize}
    \item \textbf{Construction Difference:}  Instead of constructing a greedy $(1+\eps)$-spanner as in~\cite{BLW17}, we will construct a greedy $(1+\eps)(2k-1)$-spanner. (Recall that a greedy spanner is a spanner constructed by the greedy strategy described in \Cref{sec: size}.)
    \item \textbf{Analysis Difference:} Instead of arguing that $n$-vertex  $K_h$-minor-free graphs have $O(nh\sqrt{\log h})$ edges, we will apply \Cref{lem:size} to argue that $n$-vertex $K_h$-minor-free graphs with girth $>2k$ have at most $O(n \cdot h^{\frac{2}{k+1}} \cdot \text{\normalfont polylog } h)$ edges. 
\end{itemize}

Let $G = (V, E)$ be an $n$-vertex $K_h$-minor-free weighted graph.
We may assume without loss of generality that for each edge $(u, v) \in E(G)$, $w((u, v)) = \dist_G(u, v)$ by removing edges from $G$ that do not satisfy this condition. 
Let $k \in \mathbb{Z}^+$ be a positive integer, and let $\eps \in \mathbb{R}^+$ be a sufficiently small real number. For ease of presentation, we will  construct a $(1+s\cdot \eps)(2k-1)$-spanner $H \subseteq G$, for a sufficiently large constant $s \in \mathbb{Z}^+$. This will not affect our bounds in \Cref{thm:light_spanner} because we can always decrease $\eps$ by a $\frac{1}{s}$ factor to remove the dependency on $s$ in our distortion bounds.

% and then prove it has the lightness claimed in \Cref{thm:light_spanner} in the following section. 

\paragraph{Construction of $H$.} 
We define our light spanner $H$ to be a greedy $(1+s \cdot \eps)(2k-1)$-spanner of $G$. Initially, let $H \gets (V, \emptyset)$. For each edge $(u, v) \in E$ in order of non-decreasing weight, if $\dist_H(u, v) > (1+s \cdot \eps)(2k-1) \cdot \dist_G(u, v)$, then add edge $(u, v)$ to $H$. After iterating through all edges in $E$, return $H$.

We will prove some basic properties of subgraph $H$ before proceeding with the lightness upper bound in the following section.

\begin{claim}[Correctness] \label{correctness:light}
    $H$ is a $(1+s \cdot \eps)(2k-1)$-spanner of $G$. 
\end{claim}
\begin{proof}
    Follows from an identical argument as in \Cref{clm:correct}.
\end{proof}

We now define the notion of \textit{weighted girth}, which generalizes standard notion of graph girth to the weighted setting.

\begin{definition}[Normalized Cycle Weight and Weighted Girth \cite{elkin2015light}] \label{def:wtdgirth}
    We define the normalized weight of a cycle $C$ in $G$ to be
    $$
    w^*(C) := \frac{w(C)}{\max_{e \in C}w(e)}.
    $$
 The weighted girth of $G$ is defined as the minimum normalized weight $w^*(C)$ of any cycle $C$ in graph $G$. 
\end{definition}

By a standard argument in the light spanner literature (see, e.g., \cite{elkin2015light, Bodwin25}), subgraph $H$ has large weighted girth.

\begin{claim}[cf. Lemma 3.2 of \cite{Bodwin25}]
$H$ is a $K_h$-minor-free graph with weighted girth $>(1+s \cdot \eps)\cdot 2k$.  \label{clm:minorfreeh}
\end{claim} 
\begin{proof}
    Since $G$ is $K_h$-minor-free, subgraph $H$ is $K_h$-minor-free as well. Additionally, by Lemma 3.2 of \cite{Bodwin25}, since $H$ is a greedy spanner of $G$ with approximation parameter $(1+s\cdot \eps) \cdot 2k$, subgraph $H$ has weighted girth $>(1+s \cdot \eps) \cdot 2k$, as claimed.  
\end{proof}

We note that if $H$ has weighted girth $>t$, then $H$ has girth $>t$ as well. Consequently, \Cref{clm:minorfreeh} implies that $H$ has girth $>2k$. 

\paragraph{Reduction to unit-weight MST edges.} In our analysis, we can assume without loss of generality that $H$ has a minimum spanning tree with unit weight edges using a standard reduction in the literature~\cite{CW18, BLW17, Bodwin25}.  This reduction can be summarized with the following lemma.

\begin{lemma}[cf. Lemma 3.4 of \cite{Bodwin25}]
    Let $H$  be an $n$-vertex $K_h$-minor-free graph graph with weighted girth $>t$ and lightness $\ell$. Then there exists a $K_h$-minor-free graph $H'$ with $O(n)$ vertices, weighted girth $>t$, and lightness $\Omega(\ell)$ such that $H$ has unit weight MST edges and every edge in $H$ has weight at least one.\footnote{The proof of Lemma 3.4 in \cite{Bodwin25} does not explicitly prove that $H'$ is $K_h$-minor-free. However, $H'$ is constructed by subdividing and reweighting edges in $H$, which does not introduce any new clique minors into $H'$.}  
\end{lemma}

For the remainder of this section, we assume that the MST of $H$ has weight one edges, and all edges in $H$ have weight at least one. 

\subsection{Applying the Hierarchical Clustering Analysis of \cite{BLW17}}

In this section, we will apply the analysis of \cite{BLW17} to complete the proof of \Cref{thm:light_spanner}. We begin by partitioning the edge set $E(H)$ of our spanner $H$ in the same way as  \cite{BLW17}. 

Let $\sigma_{h, k}$ denote the largest average degree of any $K_h$-minor-free graph with girth $>2k$. By \Cref{lem:size}, $\sigma_{h, k} = O(h^{\frac{2}{k+1}} \cdot \text{\normalfont polylog } h)$. 
Let $E' \subseteq E(H)$ be the edges in $H$ with weights in the range $[1, 1/\eps)$. We begin with the following simple observation.
\begin{claim}[cf. equation 3 of \cite{BLW17}]
    $w(E') = O(\sigma_{h, k} \cdot n /\eps) = O\left(\sigma_{h, k}  \cdot \eps^{-1} \cdot w(\text{\normalfont MST}) \right)$ 
\end{claim}
\begin{proof}
    We know that $w(\text{MST}) = n-1$, by our assumption that the edges of the minimum spanning tree of $G$ are unit weight. Additionally, we know that $|E'| \leq \sigma_{h, k} \cdot n$, since $H$ is a $K_h$-minor-free graph with girth $>2k$ by \Cref{clm:minorfreeh}.   
    By the definition of $E'$, we conclude that
     $w(E') = O(\sigma_{h, k} \cdot n /\eps) = O\left(\sigma_{h, k}  \cdot \eps^{-1} \cdot w(\text{\normalfont MST}) \right).$ 
\end{proof}

We now develop some of the notation and definitions used in the analysis of the light spanner in \cite{BLW17}. 
Let $E^i_j \subseteq E(H)$ be the set of edges in $H$ with weights in the range $\left[\frac{2^j}{\eps^i}, \frac{2^{j+1}}{\eps^i}\right)$ for every $i \in \mathbb{Z}^+$ and every $j \in \left[0, \left\lceil \log \frac{1}{\eps}\right\rceil \right]$. Let $J_j = \cup_{i \in \mathbb{Z}^+}E^i_j$ for each  $j \in \left[0, \left\lceil \log \frac{1}{\eps}\right\rceil \right]$.
% In the remainder of this section, we will prove  \Cref{lem:light main lemma} for a fixed $j \in \left[0, \left\lceil \log \frac{1}{\eps}\right\rceil \right]$. Towards that goal, 
Next we will  briefly develop the relevant details of the clustering analysis used in \cite{BLW17}. 

\paragraph{Clustering analysis of $J_j$.} 
Fix a $j \in \left[0, \left\lceil \log \frac{1}{\eps}\right\rceil \right]$. 
Let $E_j^0 \subseteq E(G)$ be the set of edges in the MST of $G$. 
For each $i \in \mathbb{Z}_{\geq 0}$, we define the subgraph $H_j^i = (V, \cup_{i\in \mathbb{Z}_{\geq 0}} E_j^i)$ of $H$.  Additionally, for each $i \in \mathbb{Z}_{\geq 0}$, the authors of \cite{BLW17} define a collection $\mathcal{C}_j^i$ of \textit{clusters} of graph $H_j^i$. Each cluster $C \in \mathcal{C}_j^i$ corresponds to a subgraph of $H_j^i$. Additionally, clusters in $\mathcal{C}_j^i$ are pairwise vertex-disjoint.  Let $\ell_j^i = \frac{2^{j+1}}{\eps^i}$. The clusters in $\mathcal{C}_j^i$ satisfy the following invariant.
\begin{center}
    \textbf{Low Diameter Invariant \normalfont{(cf. DC1 of \cite{BLW17})}.}  Let $g \in \mathbb{Z}^+$ be a positive integer constant.  For each $i \in \mathbb{Z}^+$, each cluster $C \in \mathcal{C}_j^i$ has diameter at most $g \ell_j^i$.\footnote{In \cite{BLW17}, $g$ is approximately 100.}
\end{center}

In addition to defining clusters $\mathcal{C}_j^i$, we also define  \textit{cluster graphs}  $\mathcal{K}_j^i$ as in \cite{BLW17}.

We define the vertex set of  $\mathcal{K}_j^i$ to be the clusters in $\mathcal{C}_j^i$, so that $V(\mathcal{K}_j^i) = \mathcal{C}_j^i$. For each edge $(u, v) \in E_j^{i+1}$ where $u \in V(C_1)$ and $v \in V(C_2)$ for clusters $C_1, C_2 \in \mathcal{C}_j^i$, we add the edge $(C_1, C_2)$ to graph $\mathcal{K}_j^i$. We make the following observation.

\begin{claim}[cf. Observation 3 of \cite{BLW17}] \label{clm : minor}
     $\mathcal{K}_j^i$ is a simple graph. Moreover, $\mathcal{K}_j^i$ is a minor of $H_j^{i+1}$. 
\end{claim}
\begin{proof}
    Every edge in $E_j^{i+1}$ has weight at least $\frac{2^j}{\eps^{i+1}} = \ell_j^i \cdot \frac{1}{2\eps}$, whereas every edge in $H_j^i$ has weight less than $\ell_j^i$. If $\eps$ is sufficiently smaller than $\frac{1}{g}$, then by the low diameter invariant, graph $\mathcal{K}_j^i$ does not contain any self-loops (since any self-loop in  $\mathcal{K}_j^i$ would contradict the fact that $H$ has weighted girth $>1+\eps$).
    
    Suppose towards contradiction that graph $\mathcal{K}_j^i$ has two parallel edges between two clusters $C_1, C_2 \in \mathcal{C}_j^i$, with corresponding edges $(u, v), (x, y) \in C_1 \times C_2$ in $E_j^{i+1}$.
      Let $w((u, v)) \leq w((x, y))$, without loss of generality.
    Let $Q = C_1 \cup C_2 + (u, v)$ be the graph obtained by unioning clusters $C_1$ and $C_2$ and the intercluster edge $(u, v)$. We observe that
\begin{align*}
    \dist_Q(x, y) &  \le \dist_Q(x, u) + w((u, v)) + \dist_Q(v, y)  & \text{ by the triangle inequality,} \\
   &   \le w((u, v)) + 2g\ell_j^i  & \text{ by the low diameter invariant,} \\
   &   \le \dist_G(x, y) + 2g\ell_j^i  & \text{ since $w((u, v)) \le w((x, y)) =\dist_G(x, y)$,} \\
   & \le (1+s \cdot \eps) \dist_G(x, y)  & \text{ by letting $s$ be sufficiently larger than $g$.}
\end{align*}
In our greedy construction of spanner $H \subseteq G$, every edge in $Q$ is added to $H$ before edge $(x, y)$ (because all edges in $Q$ have weight less than $w((x, y))$. However, the above sequence of inequalities imply that edge $(x, y)$ is not added to $H$ in our greedy $(1+s \cdot \eps)(2k-1)$-spanner algorithm. We conclude that  $\mathcal{K}_j^i$ is a simple graph.

We now finish the proof by showing that $\mathcal{K}_j^i$ is a minor of $H_j^{i+1}$. Note that clusters in $\mathcal{C}_j^i$ are connected subgraphs of $H_j^{i+1}$ by the low diameter invariant, and clusters in $\mathcal{C}_j^i$ are pairwise vertex-disjoint. Then if we contract each cluster $C$ in $H_j^{i+1}$, the resulting graph will contain $\mathcal{K}_j^i$ as a subgraph.   
\end{proof}

We now prove that the cluster graph $\mathcal{K}_j^i$ excludes a large minor and has large girth. 

\begin{claim}\label{claim: graph props}
    Graph $\mathcal{K}_j^i$ is $K_h$-minor-free and has girth $> 2k$. 
\end{claim}
\begin{proof}
    By \Cref{clm : minor}, $\mathcal{K}_j^i$  is a minor of $H_j^{i+1} \subseteq G$. Since $G$ is $K_h$-minor-free,  $\mathcal{K}_j^i$ is $K_h$-minor-free too. 
    What remains is to show that cluster graph $\mathcal{K}_j^i$ has girth $>2k$. We will need the following observations:
    \begin{itemize}
        \item Graph $H_j^{i+1}$ has weighted girth  $> (1+s \cdot \eps) \cdot 2k$ by \Cref{clm:minorfreeh}.
        \item Every cluster $C \in \mathcal{C}_j^i$ has diameter at most $\text{diam}(C) \le g\ell^i_j$ by the low diameter invariant.
        \item Every edge  $e \in E_j^{i+1}$ has weight at least $w(e) \ge \eps^{-1} / 2  \cdot \ell^i_j$ by the definition of $E_j^{i+1}$.  
    \end{itemize} 
    Now consider a cycle $Q$ in cluster graph $\mathcal{K}_j^i$. Let $C_1, \dots, C_{|Q|} \in \mathcal{C}_j^i$ denote the $|Q|$ clusters in $\mathcal{C}_j^i$ that make up the vertices of cycle $Q$ in $\mathcal{K}_j^i$.  Let $(u_1, v_1), \dots, (u_{|Q|}, v_{|Q|}) \in E_j^{i+1}$ denote the $|Q|$ edges in $E_j^{i+1}$ such that $(u_i, v_i) \in V(C_i) \times V(C_{i+1})$ for each $i \in [1, |Q|]$. Additionally, for each $i \in [1, |Q|]$, let $P_i$ be a $(v_i, u_{i+1})$-shortest path in cluster $C_i$.

    Now suppose towards contradiction that cycle $Q$ has length at most $|Q| \leq 2k$. Then we can define the following cycle $Q'$ in graph $H_j^{i+1}$:
    $$
    Q' = (u_1, v_1) \circ P_1 \circ (u_2, v_2) \circ \dots \circ P_{|Q|-1} \circ (u_{|Q|}, v_{|Q|}) \circ P_{|Q|}.
    $$
    Since each path $P_i$ has weight $w(P_i) \leq \text{diam}(C_{i+1}) \le g\ell^i_j$ and each edge $(u_i, v_i)$ has weight at least $w((u_i, v_i)) \ge \eps^{-1}/2 \cdot \ell^i_j$, the normalized weight of cycle $Q$ is at most  
    $$
    \frac{w(Q')}{\max_{e \in Q'}w(e)} \leq \frac{ |Q| \cdot \left( g\ell^i_j + \max_{i\in[1, |Q|]} w((u_i, v_i)) \right) }{\max_{i\in[1, |Q|]} w((u_i, v_i))} \leq |Q| \cdot \left( 1+ \frac{g\ell^i_j}{\eps^{-1}/2 \cdot \ell^i_j} \right) \le (1+2g \cdot \eps)2k. 
    $$
    If we choose parameter $s$ in our construction of $H$ to be $s \ge 2g$, then this contradicts \Cref{clm:minorfreeh}, which states that  $H$ has weighted girth  $> (1+s \cdot \eps) \cdot 2k$. 
\end{proof}

\paragraph{Applying the lightness argument of \cite{BLW17}.} 

For each $j \in \left[0, \left\lceil \log \frac{1}{\eps}\right\rceil \right]$, the authors of \cite{BLW17} upper bound the lightness of edges in $E_j^{i}$ for all $i \in \mathbb{Z}^+$ using a careful charging scheme.
We will black-box their charging scheme using Lemma \ref{lem:blackbox} to follow.

At level $i$, the clusters in $\mathcal{C}_j^i$ are clustered together using edges in $E_j^{i+1}$ to form the level $i+1$ clusters $\mathcal{C}_j^{i+1}$. During level $i$, the authors of \cite{BLW17} need to carefully pay for the weights of edges in $E_j^{i+1}$ using credits stored in each cluster in $\mathcal{C}_j^i$, while maintaining enough excess credits to pay for edges in higher levels as well. A key observation used in their proof is that cluster graph $\mathcal{K}_j^i$ is $K_h$-minor-free, and therefore on average each cluster in $\mathcal{C}_j^i$ is incident to $O(h \sqrt{\log h})$ edges in $\mathcal{K}_j^i$.  This charging scheme analysis yields the following lemma, which is implicit in the proof of Lemma 2 in \cite{BLW17}. 

\begin{lemma}[cf. Lemma 2 of \cite{BLW17}] \label{lem:blackbox}
    Let $d \in \mathbb{R}^+$ be the largest average degree of any cluster graph $\mathcal{K}_j^i$ for $j \in \left[0, \left\lceil \log \frac{1}{\eps}\right\rceil \right]$ and $i \in \mathbb{Z}^+$.\footnote{Although cluster graph $\mathcal{K}_i^j$ is defined for all $i \in \mathbb{Z}^+$,  there are only finitely many distinct cluster graphs in the hierarchy. Thus, $d$ is well-defined.}    There exists a set of edges $B \subseteq E(H)$ such that $w(B) = O\left( \frac{1}{\eps^2} w(\text{\normalfont MST})  \right)$ and for every  $j \in \left[0, \left\lceil \log \frac{1}{\eps}\right\rceil \right]$,
    $$
    w\left(J_j - B \right) \le O\left( \frac{d}{\eps^3} \right) w(\text{\normalfont MST}).
    $$
\end{lemma}

We can combine \Cref{lem:blackbox} with our upper bound on the sparsity of minor-free graphs with large girth in \Cref{claim: graph props} to prove our light spanner upper bound.  

\lightspanner*
\begin{proof}
 By \Cref{claim: graph props}, every cluster graph $\mathcal{K}_j^i$ is $K_h$-minor-free and has girth $>2k$. Then since $\sigma_{h, k}$ is the largest average degree of any $K_h$-minor-free graph with girth $>2k$, every cluster graph $\mathcal{K}_j^i$ has average degree at most $\sigma_{h, k}$. 
 
 Then by \Cref{lem:blackbox},  there exists a set of edges $B \subseteq E(H)$ such that
    $$
    w(H) = w(E') + w(B) +  \sum_{ j \in \left[0, \left\lceil \log \frac{1}{\eps}\right\rceil \right]} w(J_j - B) = O\left( \frac{\sigma_{h, k}}{\eps^3}  \log \frac{1}{\eps} \right) w(\text{MST}).
    $$
    Plugging in our upper bound of $\sigma_{h, k} = O\left(h^{\frac{2}{k+1}} \cdot \text{\normalfont polylog } h\right)$ on $\sigma_{h, k}$ from \Cref{lem:size} completes the proof of \Cref{thm:light_spanner}. 
\end{proof}

\bibliographystyle{alpha}
\bibliography{REF}

\end{document}